\documentclass[runningheads,11points]{llncs}

\usepackage[utf8, latin5]{inputenc}
\usepackage{amsmath,amssymb}

\usepackage{graphicx}

\usepackage{tikz}

\newtheorem{prop}{Proposition}

\newtheorem{obs}{Observation}

\newcommand{\gp}{\gamma_{\rm pr}}

\newcommand{\probl}[3]{
\vspace{-5mm}
\begin{flushleft}
\fbox{
\begin{minipage}{115mm}
\noindent {\sc #1}\\
          {\bf Input:} #2\\
          {\bf Output:} #3
\end{minipage}}
\medskip
\end{flushleft}
\vspace{-5mm}
}

\title{Paired Domination versus Domination and Packing Number in Graphs\thanks{This work is supported by the Scientific and Technological Research Council of Turkey (TUBITAK) under grant no.118E799.}
}

\author{Magda Dettlaff~\inst{1}, Didem G\"{o}z\"{u}pek~\inst{2},  and Joanna Raczek~\inst{3}
}

\institute{Faculty of Applied Physics and Mathematics, Gdansk University of Technology, Poland\\
\email{magda.dettlaff1@pg.edu.pl}
\and
Department of Computer Engineering, Gebze Technical University, Kocaeli, Turkey\\
\email{didem.gozupek@gtu.edu.tr}
\and
Faculty of Electronics, Telecommunications and Informatics, Gdansk University of Technology, Poland\\
\email{joanna.raczek@pg.edu.pl}
}

\begin{document}

\maketitle
\begin{abstract}
Given a graph $G=(V(G), E(G))$, the size of a minimum dominating set, minimum paired dominating set, and a minimum total dominating set of a graph $G$ are denoted by $\gamma(G)$, $\gamma_{\rm pr}(G)$, and $\gamma_{t}(G)$, respectively. For a positive integer $k$, a \emph{k-packing} in $G$ is a set $S \subseteq V(G)$ such that for every pair of distinct vertices $u$ and $v$ in $S$, the distance between $u$ and $v$ is at least $k+1$. The \emph{k-packing number} is the order of a largest $k$-packing and is denoted by $\rho_{k}(G)$. It is well known that $\gamma_{\rm pr}(G) \le 2\gamma(G)$. In this paper, we prove that it is NP-hard to determine whether  $\gamma_{\rm pr}(G) = 2\gamma(G)$ even for bipartite graphs. We provide a simple characterization of trees with $\gamma_{\rm pr}(G) = 2\gamma(G)$, implying a polynomial-time recognition algorithm. We also prove that even for a bipartite graph, it is NP-hard to determine whether $\gamma_{\rm pr}(G)=\gamma_{t}(G)$. We finally prove that it is both NP-hard to determine whether $\gamma_{\rm pr}(G)=2\rho_{4}(G)$ and whether $\gamma_{\rm pr}(G)=2\rho_{3}(G)$.\\

\end{abstract}

\section{Introduction}\label{intro}
The notion of paired domination was introduced by Haynes and Slater  1998 (see \cite{HSlater}) and is now widely studied in graph theory. A set $D\subseteq V(G)$ is a \emph{dominating set} of a graph $G$ if every vertex in $V(G)-D$ is adjacent to at least one vertex in $D$. The \emph{domination number} $\gamma(G)$ of a graph $G$ is the minimum cardinality of a dominating set of $G$. A set $D\subseteq V(G)$ is a \emph{paired dominating set} of a graph $G=(V,E)$ if it is a dominating set and the induced subgraph $\langle D\rangle$ has a perfect matching. The \emph{paired domination number} $\gp(G)$ is the cardinality of a smallest paired dominating set of $G$. If we think of each $u\in V(G)$ as the location of a guard capable of protecting each vertex in $N_G[u]$, then for paired domination we require the guards' locations to be selected as adjacent pairs of vertices so that every vertex of $G$ is protected and each guard is assigned to one another so that they can cooperate as backups for each other. A set $D\subseteq V(G)$ is a \emph{total dominating set} of a graph $G$ if every vertex in $V(G)$ is adjacent to at least one vertex in $D$. The \emph{total domination number} $\gamma_t(G)$ of a graph $G$ is the minimum cardinality of a total dominating set of $G$. It is shown in \cite{HHS98} that for any graph $G$ without isolated vertices, 

\begin{equation}\label{e1}
\gamma(G)\leq \gamma_t(G)\leq \gamma_{\rm pr}(G)\leq 2\gamma(G).
\end{equation}
Many authors have hitherto studied various aspects of this inequality chain. Alvarado et.al. \cite{AlvaroRautenbach} proved that it is NP-hard to decide for a given graph $G$, whether $\gamma(G)=\gamma_t (G)$, as well as it is both NP-hard to decide for a given graph $G$, whether $\gamma(G)=\gamma_{\rm pr} (G)$ and whether $\gamma_{\rm pr}(G)=2\gamma (G)$. In this paper, we complement these results by proving that it is NP-hard to decide for a given graph $G$, whether $\gamma_{\rm pr}(G)=2\gamma (G)$ even for bipartite graphs. 

Whenever two graph parameters, say $\eta(G)$ and $\tau(G)$, are related by an inequality $\eta(G) \leq\tau (G)$ for every graph $G$, it makes sense to study the class of the so-called $(\eta, \tau )$--\emph{graphs}, the graphs for which $\eta(G) = \tau (G)$. Shang, Kang and Henning have characterized the class of trees $T$ achieving $\gamma_t(T)=\gamma_{\rm pr}(T)$, see~\cite{ShangKangHenning}. In Section~\ref{NP} we prove that it is NP-hard to decide for a given graph $G$, whether $\gamma_t(G)=\gamma_{\rm pr}(G)$ even for bipartite graphs. 

Henning and Vestergaard in~\cite{HenningV} gave a constructive characterization of $(\gamma_{\rm pr},2\gamma)$-trees by showing eleven operations which are needed to construct such trees. Similarly, Hou \cite{Hou} gave a simpler characterization of such trees using three operations.  In Section~\ref{trees}, we provide a much simpler and straightforward characterization of those trees, not requiring any operations. Our characterization checks in polynomial time whether a tree $T$ is a $(\gamma_{\rm pr},2\gamma)$-tree and therefore can be used as a basis for characterizing more complex classes of $(\gamma_{\rm pr},2\gamma)$-graphs. 

For a positive integer $k$, a {\em k--packing} in $G$ is a set $A\subseteq V(G)$ such that for every pair of distinct vertices $u$ and $v$ in $A$, the distance between $u$ and $v$ in $G$ is at least $k+1$. The {\em k--packing number} is the order of a largest $k$--packing of $G$ and is denoted by $\rho_k(G)$. A 1--packing set is an \emph{independent set}, denoted also by $i(G)$. It is an easy observation, that for $k\geq 1$ and for any graph $G$, $\rho_k(G)\geq \rho_{k+1}(G)$. Bresar \cite{Bresar} proved that for any graph $G$ without isolated vertices, $\gamma_{\rm pr}(G)\geq 2\rho_3(G)$ and for every nontrivial tree $T$, $\gamma_{\rm pr}(T)= 2\rho_3(T)$. In Section~\ref{rho} we prove that it is NP-hard to decide for a given graph $G$, whether $\gp(G)=2\rho_3 (G)$, as well as it is NP-hard to decide for a given graph $G$, whether $\gp(G)=2\rho_4 (G)$.

\section{Complexity results for paired domination in bipartite graphs}\label{NP}
We start this section by  presenting the 3SAT problem.\\

\probl
{$3$-Satisfiability ($3$-SAT)}
{A boolean expression $E$ in conjunctive normal form (CNF), that is, the
conjunction of clauses, each of which is the disjunction of three distinct literals.}
{Is $E$ satisfiable?}


Now we focus on the computational complexity of the problem of determining whether total domination and paired domination numbers are equal in bipartite graphs.
\begin{theorem}\label{NPP}
It is NP-hard to determine whether $\gamma_{t}(G)=\gamma_{\rm pr}(G)$ even for a bipartite graph $G$.
\end{theorem}
\begin{proof}
We describe a reduction from 3SAT, which was proven to be \emph{NP}-complete in \cite{GJ79}, to the considered problem. The formula in 3SAT is given in conjunctive normal form, where each clause contains three literals. We assume that the formula contains the instance of any literal $u$ and its negation $\neg u$ (in the other case all clauses containing the literal $u$ are satisfied by the true assignment of $u$). 

Given an instance $E$, the set of literals $U=\{u_1,u_2,\ldots,u_n\}$ and the set of clauses $C=\{c_1,c_2,\ldots,c_m\}$ of 3SAT, we construct a graph $G$ whose order is polynomially bounded in terms of $n$ and $m$ such that the formula is satisfiable if and only if $\gamma_{t}(G)=\gamma_{\rm pr}(G)$.

For each literal $u_i$ construct a copy of the graph $G(u_i)$ in Fig.~\ref{f4}.


\begin{figure}[h!]
\begin{center}
\begin{tikzpicture}[scale=0.7]

\draw (0,0) -- (0,4);
\draw (0,2) -- (2,2);

\shade[ball color=black] (0,4) circle (3pt)node[left] {$u_i$} ;
\shade[ball color=black] (2,2) circle (3pt)node[below] {$a_i$} ;
\shade[ball color=black] (0,2) circle (3pt)node[left] {$b_i$};
\shade[ball color=black] (0,0) circle (3pt) node[left] {$u'_i$};
\end{tikzpicture}
\end{center}
\caption{The graph $G(u_i)$}\label{f4}
\end{figure}
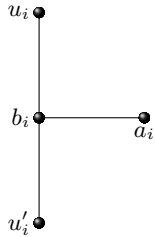

For every clause $C_j$, we create a copy of $G(C_j)$ of the graph in Fig.~\ref{f5}. All graphs $G(u_i)$ and $G(C_j)$ created so far are trees and they are disjoint. 

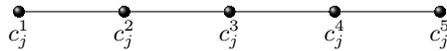
\begin{figure}[h!]
\begin{center}
\begin{tikzpicture}[scale=0.7]

\draw (0,0) -- (8,0);

\shade[ball color=black] (8,0) circle (3pt)node[below] {$c_j^5$} ;
\shade[ball color=black] (6,0) circle (3pt)node[below] {$c_j^4$} ;
\shade[ball color=black] (4,0) circle (3pt)node[below] {$c_j^3$} ;
\shade[ball color=black] (2,0) circle (3pt)node[below] {$c_j^2$};
\shade[ball color=black] (0,0) circle (3pt) node[below] {$c_j^1$};
\end{tikzpicture}
\end{center}
\caption{The graph $G(C_j)$}\label{f5}
\end{figure}

For every clause $C_j$ with literals $x, y$ and $z$, we create the three edges $c_j^5x, c_j^5y$ and $c_j^5z$. If, for example, $C_1 = \lnot u_1 \lor u_2 \lor  u_3$, then these edges are $c_1^5u'_1, c_1^5u_2$ and $c_1^5u_3$ as shown in Fig.~\ref{f6}. This completes the description of $G$. Observe that $G$ does not contain an odd cycle and therefore $G$ is bipartite.

\begin{figure}[h!]
\begin{center}
\begin{tikzpicture}[scale=0.4]

\draw (0,14) -- (0,18);
\draw (0,16) -- (2,16);

\draw (0,7) -- (0,11);
\draw (0,9) -- (2,9);

\draw (0,0) -- (0,4);
\draw (0,2) -- (2,2);

\draw(-16,9)--(-8,9)--(0,14);
\draw(0,11)--(-8,9)--(0,4);

\shade[ball color=black] (0,18) circle (5pt)node[right] {$u_1$} ;
\shade[ball color=black] (2,16) circle (5pt)node[below] {$a_1$} ;
\shade[ball color=black] (0,16) circle (5pt)node[left] {$b_1$};
\shade[ball color=black] (0,14) circle (5pt) node[right] {$u'_1$};

\shade[ball color=black] (0,11) circle (5pt)node[right] {$u_2$} ;
\shade[ball color=black] (2,9) circle (5pt)node[below] {$a_2$} ;
\shade[ball color=black] (0,9) circle (5pt)node[left] {$b_2$};
\shade[ball color=black] (0,7) circle (5pt) node[right] {$u'_2$};

\shade[ball color=black] (0,4) circle (5pt)node[right] {$u_3$} ;
\shade[ball color=black] (2,2) circle (5pt)node[below] {$a_3$} ;
\shade[ball color=black] (0,2) circle (5pt)node[left] {$b_3$};
\shade[ball color=black] (0,0) circle (5pt) node[right] {$u'_3$};

\shade[ball color=black] (-8,9) circle (5pt)node[below] {$c_1^5$} ;
\shade[ball color=black] (-10,9) circle (5pt)node[below] {$c_1^4$} ;
\shade[ball color=black] (-12,9) circle (5pt)node[below] {$c_1^3$} ;
\shade[ball color=black] (-14,9) circle (5pt)node[below] {$c_1^2$};
\shade[ball color=black] (-16,9) circle (5pt) node[below] {$c_1^1$};
\end{tikzpicture}
\end{center}
\caption{The edges between $G(C_1)$ and $G(u_1)\cup G(u_2)\cup G(u_3)$ for the clause $C_1 = \lnot u_1 \lor u_2 \lor u_3$}\label{f6}
\end{figure}
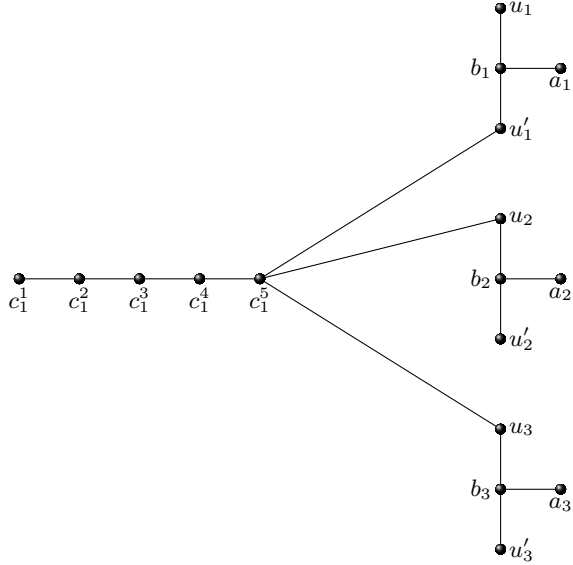

Since $c_j^2$ and $b_i$ are support vertices, they belong to every total dominating set of $G$.  However, since each of the two support vertices of $G$ are at distance at least~4, $2m+2n\leq\gamma_t(G)\leq \gamma_{\rm pr}(G)$.

First, we assume that $E$ is satisfiable and consider a satisfying truth assignment. Let $D$ be a total dominating set of $G$ such that it contains each vertex corresponding to true literal. Since $E$ is satisfiable, for each $G(C_j)$ the vertex $c_j^5$ is dominated by a vertex from $G(u_i)$. Let also $b_i, c_j^2$ and $c_j^3$ belong to $D$. This way, we construct a total dominating set $D$ of $G$ of cardinality $2m+2n$. Such a set is also a paired-dominating set of $G$, so we conclude that if $E$ is satisfiable, then $\gamma_t(G)= \gamma_{\rm pr}(G)$.

Next, we assume that $\gamma_t(G)=\gamma_{\rm pr}(G)$. Let $D$ be a minimum total dominating set of $G$. If the diameter of a graph is at least $3$, there exists a minimum total dominating set which contains no leaf. Hence, without loss of generality, we may assume that $c_j^2$ and $c_j^3$ belong to $D$. Moreover, if $|D|=2m+2n$, then from each $G(u_i)$ exactly two vertices belong to $D$: $b_i$ and either $u_i$ or $u'_i$ to dominate vertices $c_j^5$ for each $j=1,2,\dots,m$.  If the corresponding vertex of a literal belongs to $D$, then let us set the literal to true; otherwise, let us set it to false. Altogether, it follows that the truth assignment defined above satisfies $E$. Hence, assume $|D|>2m+2n$. Then no subset of $2m+2n$ vertices of $G$ totally dominates $G$, but for each $E$ there exists a subset of $2m+2n$ vertices that totally dominate $G-\{c_1^5,c_2^5,\dots,c_m^5\}$. Without loss of generality we may assume that only two vertices from each $G(C_j)$ belong to $D$, since otherwise we may exchange $c_j^4$ or $c_j^5$ with a vertex of $G(u_i)$ adjacent to $c_j^5$. However, if $D$ contains more than two vertices from a $G(u_i)$, then $D$ is not a (minimum) paired-dominating set and $\gamma_{\rm pr}(G)>|D|$, which is impossible. This completes the proof. 
\end{proof}

Our next result considers the computational complexity of the problem of determining whether paired domination number is twice the domination number in bipartite graphs.

\begin{theorem}\label{NPH}
For a given graph $G$, it is NP-hard to determine whether $\gamma_{\rm pr}(G)=2\gamma(G)$ even for bipartite graphs.
\end{theorem}
\begin{proof}
We describe a reduction from 3SAT, which was proven to be \emph{NP}-complete in \cite{GJ79}, to the considered problem. The formula in 3SAT is given in conjunctive normal form, where each clause contains three literals. We assume that the formula contains the instance of any literal $u$ and its negation $\neg u$ (in the other case all clauses containing the literal $u$ are satisfied by the true assignment of $u$). 

Given an instance $E$, the set of literals $U=\{u_1,u_2,\ldots,u_n\}$ and the set of clauses $C=\{c_1,c_2,\ldots,c_m\}$ of 3SAT, we construct a bipartite graph $G$ whose order is polynomially bounded in terms of $n$ and $m$ such that the formula is satisfiable if and only if $\gamma_{pr}(G)=2\gamma(G)$. For each literal $u_i$, we construct the graph $G(u_i)$ which is a copy of $C_6$ with vertices labeled $u_i, w_i, u'_i, z^1_i,z^2_i,z^3_i$, see Fig.~\ref{f1}.

\begin{figure}[h!]
\begin{center}
\begin{tikzpicture}[scale=0.6]

\draw (1,0) -- (0,2) -- (1,4) -- (3,4) --(4,2)--(3,0)--(1,0);

\shade[ball color=black] (1,0) circle (3pt)node[left] {$u_i$} ;
\shade[ball color=black] (1,4) circle (3pt)node[left] {$u'_i$};
\shade[ball color=black] (0,2) circle (3pt) node[right] {$w_i$} ;
\shade[ball color=black] (3,0) circle (3pt) node[right] {$z^3_i$} ;
\shade[ball color=black] (4,2) circle (3pt) node[right] {$z^2_i$} ;
\shade[ball color=black] (3,4) circle (3pt) node[right] {$z^1_i$} ;
\end{tikzpicture}
\end{center}
\caption{The graph $G(u_i)$}\label{f1}
\end{figure}
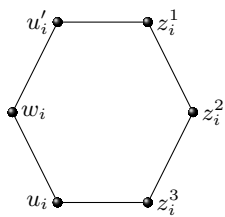

For every clause $C_j$, we create a vertex $c_j$ and for every clause $C_j$ with literals $x, y$ and $z$, we create the three edges $c_jx, c_jy$ and $c_jz$. If, for example, $C_1 = u_1 \lor u_2 \lor \lnot u_3$, then these edges are $c_1u_1, c_1u_2$ and $c_1 u'_3$ as shown in Fig.~\ref{f3}. This completes the description of $G$. By the construction of $G$ it is clear that the obtained graph is bipartite.

\begin{figure}[h!]
\begin{center}
\begin{tikzpicture}[scale=0.5]

\draw (1,0) -- (0,2) -- (1,4) -- (3,4) --(4,2)--(3,0)--(1,0);
\draw (1,6) -- (0,8) -- (1,10) -- (3,10) --(4,8)--(3,6)--(1,6);
\draw (1,12) -- (0,14) -- (1,16) -- (3,16) --(4,14)--(3,12)--(1,12);
\draw(1,0)--(-8,8)--(1,6)--(-8,8)--(1,16);

\shade[ball color=black] (1,0) circle (3pt)node[left] {$u_1$} ;
\shade[ball color=black] (1,4) circle (3pt)node[left] {$u'_1$};
\shade[ball color=black] (0,2) circle (3pt) node[right] {$w_1$} ;
\shade[ball color=black] (3,0) circle (3pt) node[right] {$z^3_1$} ;
\shade[ball color=black] (4,2) circle (3pt) node[right] {$z^2_1$} ;
\shade[ball color=black] (3,4) circle (3pt) node[right] {$z^1_1$} ;

\shade[ball color=black] (1,6) circle (3pt)node[below] {$u_2$} ;
\shade[ball color=black] (1,10) circle (3pt)node[left] {$u'_2$};
\shade[ball color=black] (0,8) circle (3pt) node[right] {$w_2$} ;
\shade[ball color=black] (3,6) circle (3pt) node[right] {$z^3_2$} ;
\shade[ball color=black] (4,8) circle (3pt) node[right] {$z^2_2$} ;
\shade[ball color=black] (3,10) circle (3pt) node[right] {$z^1_2$} ;

\shade[ball color=black] (1,12) circle (3pt)node[left] {$u_3$} ;
\shade[ball color=black] (1,16) circle (3pt)node[above] {$u'_3$};
\shade[ball color=black] (0,14) circle (3pt) node[right] {$w_3$} ;
\shade[ball color=black] (3,12) circle (3pt) node[right] {$z^3_3$} ;
\shade[ball color=black] (4,14) circle (3pt) node[right] {$z^2_3$} ;
\shade[ball color=black] (3,16) circle (3pt) node[right] {$z^1_3$} ;

\shade[ball color=black] (-8,8) circle (3pt) node[left] {$c_1$} ;

\end{tikzpicture}
\end{center}
\caption{The edges between $c_1$ and $G(u_1)\cup G(u_2)\cup G(u_3)$ for the clause $C_1 = u_1 \lor u_2 \lor \lnot u_3$}\label{f3}
\end{figure}
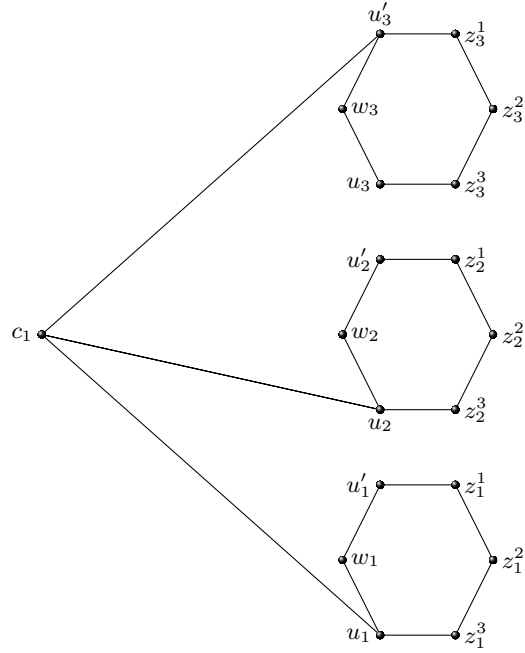

It is easy to verify that every dominating set has to include at least two vertices from each $G(u_i)$. Therefore, $\gamma(G)\geq 2n$. 

Let $D_p$ be a minimum paired dominating set of $G$. Since only two vertices of each induced cycle of $G(u_i)$ may have a neighbor outside the cycle and these two vertices are not adjacent, $|V(G(u_i))\cap D_p|\geq 3$. Moreover, no two vertices belonging to two different $G(u_i)$ components are adjacent. Therefore, $\gamma_{\rm pr}(G)\geq 4n$. On the other hand, it is easy to see that there exists a paired dominating set of size $4n$ containing $u_i$ and $u_i'$ for $i=1,2,\dots, n$, so for these reasons, $4n=\gamma_{\rm pr}(G)$.
Hence, it remains to prove that $E$ is satisfiable if and only if $\gamma(G)=2n$. 

First, we assume that $E$ is satisfiable and consider a satisfying truth assignment. We will now construct a dominating set $D$ of $G$ with size $2n$. We will first put to $D$ the vertices corresponding to each true literal. If $u_i\in D$, then we also put $z_i^1$ to $D$ and if $u'_i\in D$, then we also put $z_i^3$ to $D$. Since $E$ is satisfiable, each $c_j$ is dominated by a vertex from $G(u_i)$. Hence, this way, we have constructed a dominating set $D$ of $G$ of cardinality $2n$.

Next, assume that $\gamma(G)=2n$. Let $D$ be a minimum dominating set of $G$. Then exactly two vertices of $G(u_i)$, $i=1,2,\dots, n$ belong to $D$. Then, by the construction of $G$, each $c_j$ for $j=1,2,\dots,m$ is dominated by a vertex of $G(u_i)$, in particular, either $u_i$ or $u'_i$. For every literal, if its corresponding vertex belongs to $D$, set the literal to true; otherwise, set it to false. Altogether, it follows that the truth assignment defined above satisfies $E$. This completes the proof. 

\end{proof}

\section{Trees with Large Paired Domination Number}\label{trees}
In this section we give a characterization of ($\gamma_{\rm pr},2\gamma$)--trees.

Let $A\subseteq V(G)$ and let $v\in V(G)$. The \emph{distance} between $v$ and $A$ is the minimum distance between $v$ and a vertex of $A$.
For a tree $T$, denote by $S(T)$ the set of all support vertices of $T$ and let $R(T)$ be the set of vertices of $T$ that are at distance at least three from $S(T)$.

Before we present the main result of this section, we recall some results. 
\begin{prop}{\rm \cite{Rall}}\label{p1}
For any graph $G$ without isolated vertices, $\gamma(G)\geq \rho_2(G)$.
\end{prop}
\begin{prop}{\rm \cite{Moon}}\label{p2}
For every tree $T$, $\gamma(T)= \rho_2(T)$.
\end{prop}
\begin{prop}{\rm \cite{Bresar}}\label{p4}
For every nontrivial tree $T$, $\gamma_{\rm pr}(T)= 2\rho_3(T)$.
\end{prop}
\begin{prop}{\rm \cite{HSlater}}\label{p5}
If $\gamma_{\rm pr}(G)=2\gamma(G)$, then every minimum dominating set of $G$ is a minimum independent dominating set of $G$.
\end{prop}

An immediate consequence of the last result is what follows.
\begin{obs}
If $G$ is a ($\gamma_{\rm pr},2\gamma$)-graph and $D$ is a minimum dominating set of $G$, then choosing for every vertex of $D$ one vertex from its neighbor results in a minimum paired-dominating set.
\end{obs}

Now we state the main result of this section.
\begin{theorem}\label{th-trees}
A tree $T$ is a $(\gamma_{\rm pr},2\gamma)$--tree if and only if 
\begin{itemize}
    \item $S(T)$ is an independent set,
    \item $R(T)$ is a 3-packing in $T$, and
    \item $S(T)\cup R(T)$ is a dominating set of $T$.
\end{itemize} 
\end{theorem}

\begin{proof}
Assume first that $S(T)$ is an independent set, $R(T)$ is a 3-packing in a tree $T$ and $S(T)\cup R(T)$ is a dominating set in $T$. Then since $S(T)\cup R(T)$ is a dominating set in $T$, $\gamma(T)\leq |S(T)\cup R(T)|$. Moreover, $S(T)\cup R(T)$ is an independent set; hence, $i(T)\leq |S(T)\cup R(T)|$. For each vertex $x$ of $S(T)$ choose any leaf adjacent to $x$ and denote such a set of leaves by $S'(T)$. Thus $|S(T)| = |S'(T)|$. Since $S'(T)\cup R(T)$ is a 3-packing, $|S(T)\cup R(T)|\leq \rho_3(T)$. Therefore, since $ \rho_{k+1}(G)\leq \rho_k(G)$ and by Proposition~\ref{p2},
\[
\gamma(T)\leq |S(T)\cup R(T)|\leq \rho_3(T)\leq \rho_2(T)\leq  \gamma(T).
\]
Hence, all the above inequalities are indeed satisfied by equality. By Proposition~\ref{p4}, $\rho_3(T) = \frac 12\gamma_{pr}(T)$. Hence, $2\gamma(T) = \gamma_{pr}(T)$ and we conclude that $T$ is a $(\gamma_{pr},2\gamma)$--tree.

Now we prove that if $T$ is a $(\gamma_{pr},2\gamma)$--tree, then $S(T)$ is an independent set, $R(T)$ is a 3-packing, and $S(T)\cup R(T)$ is a dominating set of $T$.

By Proposition~\ref{p5}, if $T$ is a $(\gamma_{pr},2\gamma)$--tree, then each minimum dominating set of $T$ is independent. Since in any non-trivial tree there exists a minimum dominating set containing $S(T)$, we conclude that in every $(\gamma_{pr},2\gamma)$--tree $S(T)$ is independent. In what follows, we only consider minimum dominating sets not containing leaves.

Denote by $S_1(T)$ the set of neighbours of $S(T)$ which are not leaves and denote by $S_k(T)$ the set of vertices of distance~$k$ from $S(T)$, where $k \ge 2$. 

If $S(T)$ is an independent dominating set in $T$, then the statement is obviously true. Therefore, we assume that there exists a vertex which is not dominated by $S(T)$. Suppose that $v\in S_2(T)$ belongs to a minimum dominating set denoted by $D$. Now construct a set $P$ as follows. Root $T$ at $v$. For each $u\in D-\{v\}$ add to $P$ both $u$ and its parent in the rooted tree. If two vertices share the same parent, take the parent and any other vertex adjacent to one of the two vertices. Since two vertices in $D$ that are sharing the same parent cannot both be leaves, this is always possible. Clearly, $P$ is a dominating set and $v\notin P$. Moreover, each vertex of $P$ is paired with another vertex of $P$. Hence, $P$ is a paired dominating set of $T$. Then,
\[
\gamma_{\rm pr}(T)\leq |P| = 2|D|-2=2\gamma(T)-2<\gamma_{\rm pr}(T),
\]
 a contradiction. Therefore, no vertex of $S_2(T)$ belongs to a minimum dominating set $D$ of $T$. Hence, if $S_2(T)\neq \emptyset$, each vertex of $S_2(T)$ is adjacent to a vertex of $S_3(T)$, since otherwise $D$ is not a dominating set of $T$. Since $S_3(T) \subseteq R(T)$, we conclude that $S(T)\cup R(T)$ is a dominating set of $T$. 

Now we prove that if $T$ is a ($\gamma_{pr},2\gamma$)-tree, then $R(T)$ is a 3-packing. By Proposition \ref{p5}, two adjacent vertices of $R(T)$ cannot both be in the same minimum dominating set.


Assume $x$ and $y$ are vertices of $R(T)$ such that $(x,u,v,y)$ is a path and $x,y$ belong to a minimum dominating set of $T$, say $D$. Then $u,v\notin D$. Denote by $T_x$ the subtree of $T-xu$ containing $x$ rooted in $x$ and denote by $T_y$ the subtree of $T-yv$ containing $y$ rooted in $y$. Let $D_x=D\cap V(T_x)$ and $D_y=D\cap V(T_y)$. Then $D=D_x\cup D_y\cup D_r$, where $D_r$ are the vertices belonging to $D$ from the subtress rooted in $u$ and $v$. Let $D^p_x$ be a paired dominating set of $T_x-x$ obtained from $D_x$ by adding to $D_x$ a parent  of each element except of the parent of $x$ (there is no parent of $x$ in $T_x$ because $x$ is the root in $T_x$). Define $D^p_y$ analogously. Let $D^p_r$ be a paired dominating set of $T-(V(T_x)\cup V(T_y))$ of cardinality $2|D_r|$. Then $D^p_x\cup D^p_y\cup D^p_r \cup \{u,v\}$ is a paired dominating set of $T$ of cardinality 
\[
2(|D_x|-1) + 2(|D_y|-1) +2|D_r|+2 < 2|D|,
\]
 a contradiction. By similar arguments we again obtain a contradiction when $x$ and $y$ are vertices of $R(T)$ such that $(x,u,y)$ is a path and $x,y\in D$. Henceforth, if $x$ and $y$ are vertices of $R(T)$, then the distance between $x$ and $y$ is at least~4, that is, $R(T)$ is a 3-packing.

\end{proof}

For any tree $T$, constructing an $S(T)$-set and $R(T)$-set as well as checking whether $S(T)$ is independent, $R(T)$ is a 3-packing and $S(T)\cup R(T)$ is a dominating set can all be easily done by an algorithm with polynomial time complexity. Hence Theorem~\ref{th-trees} implies a polynomial-time recognition algorithm of trees with $\gamma_{pr}(G)=2\gamma(G)$.

\section{Packing number and paired-domination}\label{rho}
In this section, we focus on the relation between the packing number and the paired-domination number. First we start by proving that it is NP-hard to decide for a given graph $G$, whether $\gp(G)=2\rho_4 (G)$. For clarity, here we recall the result of Bresar et.al.:
\begin{prop}{\rm \cite{Bresar}}\label{p3}
For any graph $G$ without isolated vertices, $\gamma_{\rm pr}(G)\geq 2\rho_3(G)$.
\end{prop}

\begin{theorem}\label{NPPD}
It is NP-hard for a given graph $G$ to determine whether $\gamma_{\rm pr}(G)=2\rho_4(G)$.
\end{theorem}
\begin{proof}
We describe a reduction from NAE3SAT, which was proven to be \emph{NP}-complete in \cite{GJ79}, to the considered problem. The formula in NAE3SAT is given in conjunctive normal form, where each clause contains three literals. 

Given an instance $E$, the set of literals $U=\{u_1,u_2,\ldots,u_n\}$ and the set of clauses $C=\{c_1,c_2,\ldots,c_m\}$ of NAE3SAT, we construct a graph $G$ whose order is polynomially bounded in terms of $n$ and $m$ such that the three values in each clause are not all equal to each other if and only if $\gamma_{\rm pr}(G) =2\rho_4(G)$.

For each literal $u_i$ we construct a copy of the graph $G(u_i)$, as in Fig.~\ref{fpd4}. 
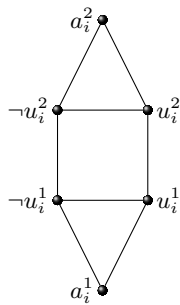
\begin{figure}[h!]
\begin{center}
\begin{tikzpicture}[scale=0.6]

\draw (0,2)--(0,4)--(1,6)--(2,4)--(2,2)--(1,0)--(0,2)--(2,2);
\draw (0,4) -- (2,4);

\shade[ball color=black] (1, 0) circle (3pt)node[left] {$a_i^1$} ;
\shade[ball color=black] (1,6) circle (3pt)node[left] {$a_i^2$} ;
\shade[ball color=black] (0,2) circle (3pt)node[left] {$\lnot u_i^1$};
\shade[ball color=black] (0,4) circle (3pt) node[left] {$\lnot u_i^2$};
\shade[ball color=black] (2,2) circle (3pt)node[right] {$u_i^1$};
\shade[ball color=black] (2,4) circle (3pt) node[right] {$u_i^2$};
\end{tikzpicture}
\end{center}
\caption{The graph $G(u_i)$}\label{fpd4}
\end{figure}

For every clause $C_j$, we create a copy $G(C_j)$ of the graph in Fig.~\ref{fpd5}. All graphs $G(u_i)$ and $G(C_j)$ created so far are disjoint. 

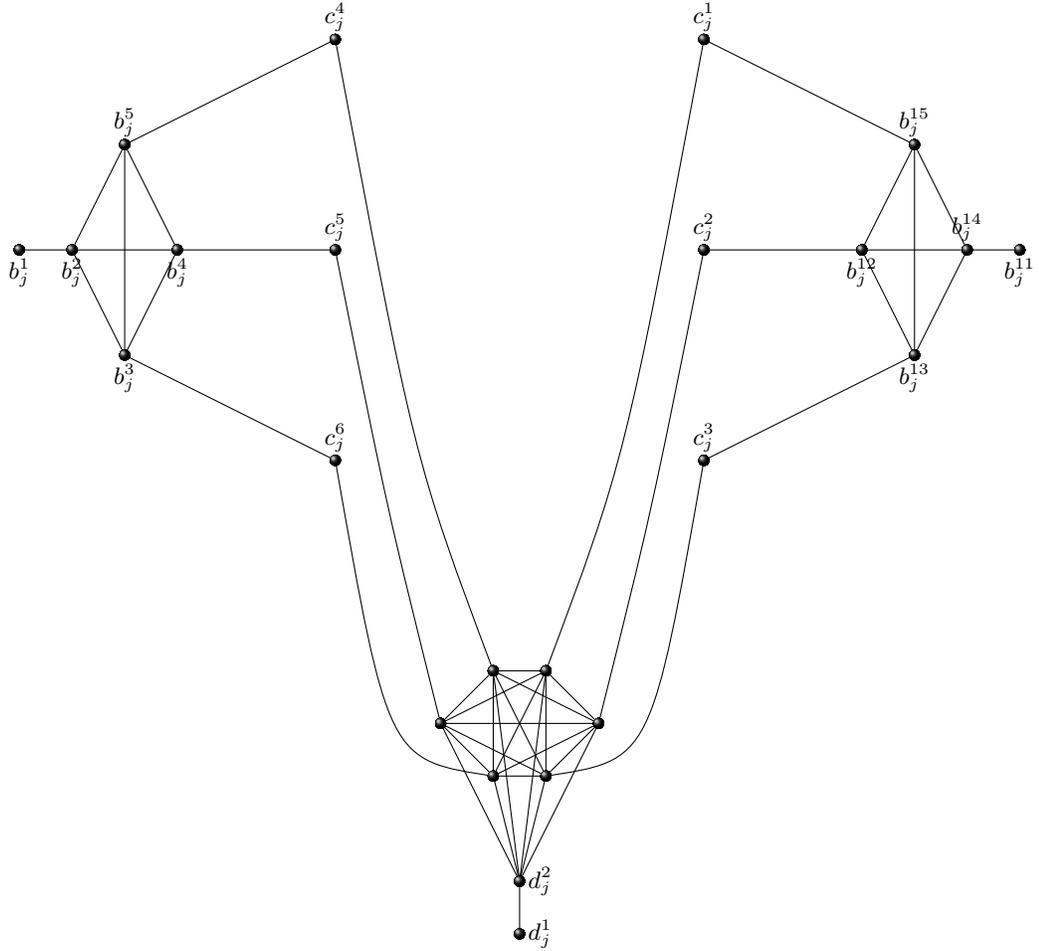
\begin{figure}[h!]
\begin{center}
\begin{tikzpicture}[scale=0.7]
\draw (6,14) .. controls(7.5,6)..  (9,2);
\draw (6,10) .. controls(7,5)..  (8,1);
\draw (6,6) .. controls(7,0.3)..  (9,0);

\draw (13,14) .. controls(11.5,6)..  (10,2);
\draw (13,10) .. controls(12,5)..  (11,1);
\draw (13,6) .. controls(12,0.3)..  (10,0);

\draw(6,14)--(2,12)--(3,10)--(2,8)--(2,12)--(1,10)--(0,10);
\draw(3,10)--(1,10)--(2,8);
\draw(3,10)--(6,10);
\draw(6,6)--(2,8);

\draw(13,14)--(17,12)--(18,10)--(17,8)--(16,10)--(19,10);
\draw(17,12)--(17,8)--(13,6);
\draw(17,12)--(16,10)--(13,10);

\draw(8,1)--(9,0)--(9,2)--(10,0)--(10,2)--(11,1)--(8,1)--(9,2)--(10,2)--(8,1)--(10,0)--(11,1)--(9,0)--(10,0);
\draw(9,2)--(11,1);
\draw(10,2)--(9,0);

\draw(9,0)--(9.5,-2)--(8,1);
\draw(10,0)--(9.5,-2)--(10,2);
\draw(9,2)--(9.5,-2)--(11,1);
\draw(9.5,-2)--(9.5,-3);

\shade[ball color=black] (6,14) circle (3pt)node[above] {$c_j^4$} ;
\shade[ball color=black] (6,10) circle (3pt)node[above] {$c_j^5$} ;
\shade[ball color=black] (6,6) circle (3pt)node[above] {$c_j^6$} ;

\shade[ball color=black] (13,14) circle (3pt)node[above] {$c_j^1$} ;
\shade[ball color=black] (13,10) circle (3pt)node[above] {$c_j^2$} ;
\shade[ball color=black] (13,6) circle (3pt)node[above] {$c_j^3$} ;

\shade[ball color=black] (0,10) circle (3pt)node[below] {$b_j^1$} ;
\shade[ball color=black] (1,10) circle (3pt)node[below] {$b_j^2$} ;
\shade[ball color=black] (2,8) circle (3pt)node[below] {$b_j^3$} ;
\shade[ball color=black] (3,10) circle (3pt)node[below] {$b_j^4$};
\shade[ball color=black] (2,12) circle (3pt) node[above] {$b_j^5$};

\shade[ball color=black] (19,10) circle (3pt)node[below] {$b_j^{11}$} ;
\shade[ball color=black] (16,10) circle (3pt)node[below] {$b_j^{12}$} ;
\shade[ball color=black] (17,8) circle (3pt)node[below] {$b_j^{13}$} ;
\shade[ball color=black] (18,10) circle (3pt)node[above] {$b_j^{14}$};
\shade[ball color=black] (17,12) circle (3pt) node[above] {$b_j^{15}$};

\shade[ball color=black] (8,1) circle (3pt) ;
\shade[ball color=black] (9,0) circle (3pt) ;
\shade[ball color=black] (9,2) circle (3pt) ;
\shade[ball color=black] (10,0) circle (3pt) ;
\shade[ball color=black] (10,2) circle (3pt) ;
\shade[ball color=black] (11,1) circle (3pt) ;

\shade[ball color=black] (9.5,-2) circle (3pt) node[right] {$d_j^2$};
\shade[ball color=black] (9.5,-3) circle (3pt) node[right] {$d_j^1$};

\end{tikzpicture}
\end{center}
\caption{The graph $G(C_j)$}\label{fpd5}
\end{figure}

For every clause $C_j$ with literals $x, y$ and $z$, we create the six edges $c_j^1x^1, c_j^1x^2,$ 
$c_j^2y^1, c_j^2y^2, c_j^3z^1$ and $c_j^3z^2$ and another six edges $c_j^4\lnot x^1, c_j^4\lnot x^2, c_j^5\lnot y^1,$ $c_j^5\lnot y^2, c_j^6\lnot z^1$ and $c_j^6\lnot z^2$. If, for example, $C_1 = \lnot u_1 \lor u_2 \lor  u_3$, then these edges are $c_1^1\lnot u_1^1, c_1^1\lnot u_1^2,$ $c_1^2u_2^1, c_1^2u_2^2, c_1^3u_3^1, c_1^3u_3^2, c_1^4 u_1^1, c_1^4 u_1^2, c_1^5\lnot u_2^1, c_1^5\lnot u_2^2, c_1^6\lnot u_3^1$ and $c_1^6\lnot u_3^2$  as shown in Fig.~\ref{fpd6}. This completes the description of $G$. The obtained graph has $6n+24m$ vertices and $8n+54m$ edges.

\begin{figure}[h!]
\begin{center}
\begin{tikzpicture}[scale=0.7]
\draw (6,14) .. controls(7.5,6)..  (9,2);
\draw (6,10) .. controls(7,5)..  (8,1);
\draw (6,6) .. controls(7,0.3)..  (9,0);

\draw (13,14) .. controls(11.5,6)..  (10,2);
\draw (13,10) .. controls(12,5)..  (11,1);
\draw (13,6) .. controls(12,0.3)..  (10,0);

\draw(6,14)--(2,12)--(3,10)--(2,8)--(2,12)--(1,10)--(0,10);
\draw(3,10)--(1,10)--(2,8);
\draw(3,10)--(6,10);
\draw(6,6)--(2,8);

\draw(13,14)--(17,12)--(18,10)--(17,8)--(16,10)--(19,10);
\draw(17,12)--(17,8)--(13,6);
\draw(17,12)--(16,10)--(13,10);

\draw(8,1)--(9,0)--(9,2)--(10,0)--(10,2)--(11,1)--(8,1)--(9,2)--(10,2)--(8,1)--(10,0)--(11,1)--(9,0)--(10,0);
\draw(9,2)--(11,1);
\draw(10,2)--(9,0);

\draw(9,0)--(9.5,-2)--(8,1);
\draw(10,0)--(9.5,-2)--(10,2);
\draw(9,2)--(9.5,-2)--(11,1);
\draw(9.5,-2)--(9.5,-3);

\draw(10,16)--(9.5,17)--(9,16)--(9,15)--(9.5,14)--(10,15)--(10,16)--(9,16);
\draw(10,15)--(9,15)--(6,14)--(9,16);
\draw(10,16)--(13,14)--(10,15);

\draw(10,11)--(9.5,12)--(9,11)--(9,10)--(9.5,9)--(10,10)--(10,11)--(9,11);
\draw(10,10)--(9,10)--(6,10)--(9,11);
\draw(10,11)--(13,10)--(10,10);

\draw(10,6)--(9.5,7)--(9,6)--(9,5)--(9.5,4)--(10,5)--(10,6)--(9,6);
\draw(10,5)--(9,5)--(6,6)--(9,6);
\draw(10,6)--(13,6)--(10,5);

\shade[ball color=black] (9.5,17) circle (3pt) node[above] {$a_1^2$};
\shade[ball color=black] (9,16) circle (3pt)node[above] {$u_1^2$} ;
\shade[ball color=black] (10,16) circle (3pt)node[above] {$\lnot u_1^2$} ;
\shade[ball color=black] (9,15) circle (3pt)node[below] {$u_1^1$} ;
\shade[ball color=black] (10,15) circle (3pt)node[below] {$\lnot u_1^1$} ;
\shade[ball color=black] (9.5,14) circle (3pt) node[below] {$a_1^1$};

\shade[ball color=black] (9.5,12) circle (3pt) node[above] {$a_2^2$};
\shade[ball color=black] (9,11) circle (3pt)node[above] {$\lnot u_2^2$} ;
\shade[ball color=black] (10,11) circle (3pt)node[above] {$u_2^2$} ;
\shade[ball color=black] (9,10) circle (3pt)node[below] {$\lnot u_2^1$} ;
\shade[ball color=black] (10,10) circle (3pt)node[below] {$u_2^1$} ;
\shade[ball color=black] (9.5,9) circle (3pt) node[below] {$a_2^1$};

\shade[ball color=black] (9.5,7) circle (3pt) node[above] {$a_3^2$};
\shade[ball color=black] (9,6) circle (3pt)node[above] {$\lnot u_3^2$} ;
\shade[ball color=black] (10,6) circle (3pt)node[above] {$ u_3^2$} ;
\shade[ball color=black] (9,5) circle (3pt)node[below] {$\lnot u_3^1$} ;
\shade[ball color=black] (10,5) circle (3pt)node[below] {$u_3^1$} ;
\shade[ball color=black] (9.5,4) circle (3pt) node[below] {$a_3^1$};

\shade[ball color=black] (6,14) circle (3pt)node[above] {$c_j^4$} ;
\shade[ball color=black] (6,10) circle (3pt)node[above] {$c_j^5$} ;
\shade[ball color=black] (6,6) circle (3pt)node[above] {$c_j^6$} ;

\shade[ball color=black] (13,14) circle (3pt)node[above] {$c_j^1$} ;
\shade[ball color=black] (13,10) circle (3pt)node[above] {$c_j^2$} ;
\shade[ball color=black] (13,6) circle (3pt)node[above] {$c_j^3$} ;

\shade[ball color=black] (0,10) circle (3pt)node[below] {$b_j^1$} ;
\shade[ball color=black] (1,10) circle (3pt)node[below] {$b_j^2$} ;
\shade[ball color=black] (2,8) circle (3pt)node[below] {$b_j^3$} ;
\shade[ball color=black] (3,10) circle (3pt)node[below] {$b_j^4$};
\shade[ball color=black] (2,12) circle (3pt) node[above] {$b_j^5$};

\shade[ball color=black] (19,10) circle (3pt)node[below] {$b_j^{11}$} ;
\shade[ball color=black] (16,10) circle (3pt)node[below] {$b_j^{12}$} ;
\shade[ball color=black] (17,8) circle (3pt)node[below] {$b_j^{13}$} ;
\shade[ball color=black] (18,10) circle (3pt)node[above] {$b_j^{14}$};
\shade[ball color=black] (17,12) circle (3pt) node[above] {$b_j^{15}$};

\shade[ball color=black] (8,1) circle (3pt) ;
\shade[ball color=black] (9,0) circle (3pt) ;
\shade[ball color=black] (9,2) circle (3pt) ;
\shade[ball color=black] (10,0) circle (3pt) ;
\shade[ball color=black] (10,2) circle (3pt) ;
\shade[ball color=black] (11,1) circle (3pt) ;

\shade[ball color=black] (9.5,-2) circle (3pt) node[right] {$d_j^2$};
\shade[ball color=black] (9.5,-3) circle (3pt) node[right] {$d_j^1$};

\end{tikzpicture}
\end{center}
\caption{The edges between $G(C_1)$ and $G(u_1)\cup G(u_2)\cup G(u_3)$ for the clause $C_1 = \lnot u_1 \lor u_2 \lor u_3$}\label{fpd6}
\end{figure}
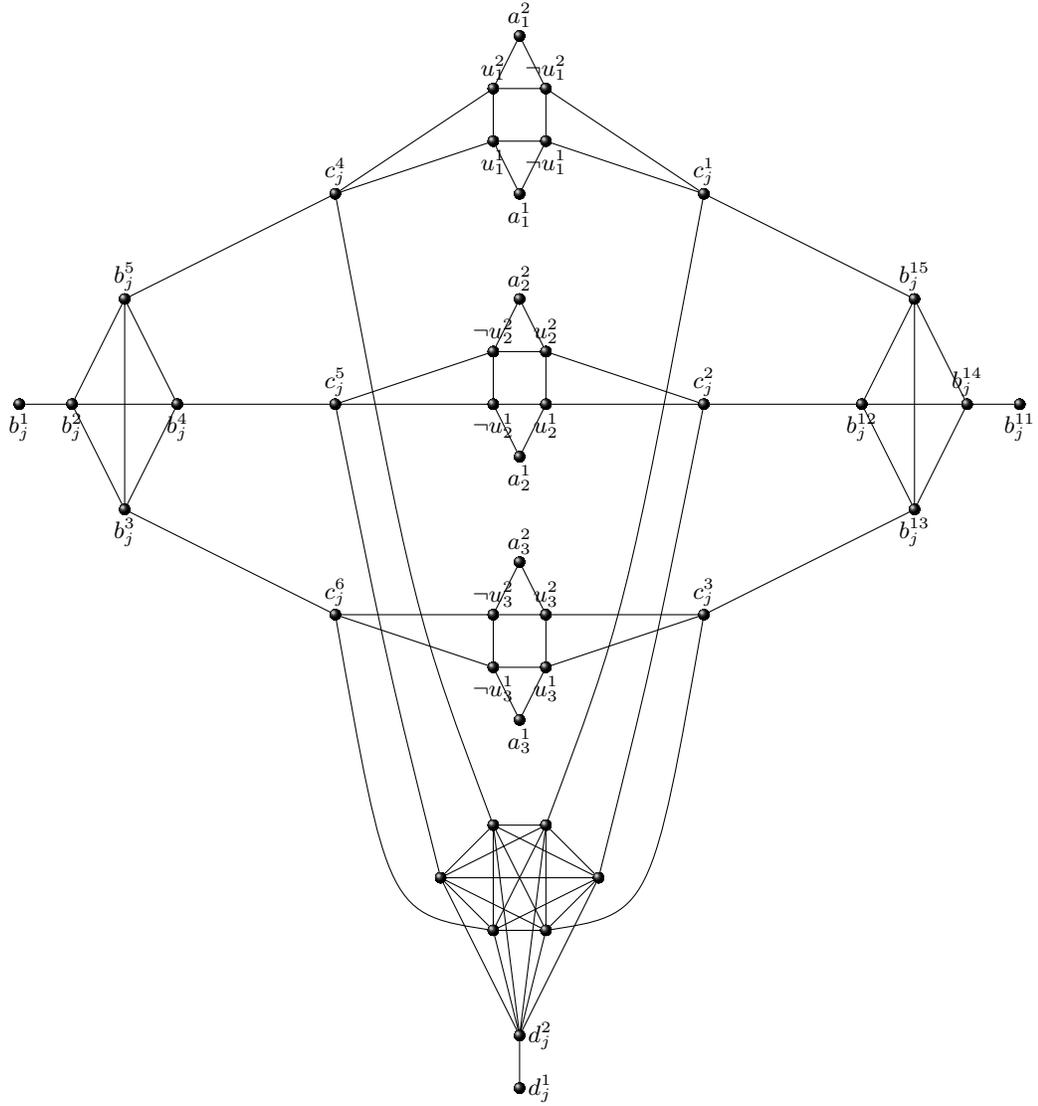

Let $R$ be a maximum $\rho_4(G)$-set. Without loss of generality we may assume that each vertex of degree~1 belongs to $R$. Then no other vertex of $G(C_j)$ for $j=1,2,\dots , m$ belongs to $R$. However, then either $a_i^1\in R$ or $a_i^2\in R$ for $i=1,2,\dots, n$. Therefore, $\rho_4(G)\geq n+3m$ and by Proposition~\ref{p3}, $\gamma_{\rm pr}(G)\geq 2\rho_3(G)\geq 2\rho_4(G)\geq 2n+6m$. On the other hand, by the construction of $G$ it is not possible to obtain a 3-packing set of cardinality greater than $n+3m$. Therefore, $\rho_4(G)=n+3m$.

First, assume that there exists an instance $E$ such that the three values in each clause are not all equal to each other and consider a satisfying truth assignment. By the construction of $G(u_i)$, every minimum paired dominating set of $G$ contains at least two vertices of this
subgraph. Let $D$ be a paired-dominating set of $G$ such that it contains vertices of $G(u_i)$ that correspond to true literals. For example, if $u_1$ is true in $E$, then $u_1^1, u_1^2\in D$ and if $\lnot u_1$ is true, then $\lnot u_1^1, \lnot u_1^2\in D$. Since the three values in each clause are not all equal to each other, at least one vertex of $c_j^1, c_j^2, c_j^3$ is adjacent to a false literal, and by symmetry the same is true for $c_j^4, c_j^5, c_j^6$. Moreover, if two vertices of $c_j^1, c_j^2, c_j^3$ are not dominated by a true literal vertex, then exactly one vertex of $c_j^4, c_j^5, c_j^6$ is not dominated by a true literal vertex, and vice versa. Without loss of generality we assume $c_j^1, c_j^3$ and $c_j^5$ are not dominated by a true literal vertex. Then let $b_j^{13}, b_j^{14}\in D$, $b_j^2, b_j^4\in D$ and $d_j^2, d\in D$, where $d$ is a vertex of $K_6$ adjacent to $c_j^1$. Then $D$ is a paired dominating set of cardinality $2n+6m$ and we conclude that $D$ is a minimum paired dominating set of $G$. Therefore, if there exists an instance such that the three values in each clause are not all equal to each other, then $2\rho_4(G)=\gamma_{\rm pr}(G)$.

Assume now that in each instance $E$ there exists a clause in which all three values are equal to each other. Let $C_k$ be such a clause and $D$ be a minimum paired-dominating set of $G$. Then at least two vertices of each $G(u_i)$ belong to $D$. Let $D$ contain vertices of $G(u_i)$ corresponding to true literals. Then either $c_k^1, c_k^2, c_k^3$ are all dominated by true literal vertices and each of $c_k^4, c_k^5, c_k^6$ is not or vice versa. In both cases, $b_k^{14}$ belongs to $D$ and it can be paired in $D$ with any neighbour. Similarly, $b_k^2\in D$ and without loss of generality let it be paired with $b_k^4$ in $D$. Next let $d_k^2\in D$ be paired with a vertex of $K_6$ adjacent to $c_k^6$. In this situation we still need two vertices to dominate $c_k^4$. This implies that $|D|> 2n+6m$, implying that if in each instance $E$ there exists a clause in which all three values are equal to each other, then $2\rho_4(G)<\gamma_{\rm pr}(G)$.

\end{proof}

Now we prove a similar result for $\rho_3$:

\begin{theorem}\label{NPPD3}
It is NP-hard for a given graph $G$ to determine whether $\gamma_{\rm pr}(G)=2\rho_3(G)$.
\end{theorem}
\begin{proof}
We describe a reduction from NAE3SAT to the considered problem. Given an instance $E$, the set of literals $U=\{u_1,u_2,\ldots,u_n\}$ and the set of clauses $C=\{c_1,c_2,\ldots,c_m\}$ of NAE3SAT, we construct a graph $G$ whose order is polynomially bounded in terms of $n$ and $m$ such that the three values in each clause are not all equal to each other if and only if $\gamma_{\rm pr}(G) =2\rho_3(G)$.

For each literal $u_i$, we construct a copy of the graph $G(u_i)$, as in Fig.~\ref{fpd4}. For every clause $C_j$, we create vertices $v_j$ and $w_j$. For every clause $C_j$ with literals $x, y$ and $z$, we create the six edges $v_jx^1, v_jx^2, v_jy^1, v_jy^2, v_jz^1$ and $v_jz^2$ and another six edges $w_j\lnot x^1, w_j\lnot x^2, w_j\lnot y^1, w_j\lnot y^2, w_j\lnot z^1$ and $w_j^3\lnot z^2$. If, for example, $C_1 = \lnot u_1 \lor u_2 \lor  u_3$, then these edges are $v_1u_1^1, v_1u_1^2, v_1u_2^1, v_1u_2^2, v_1u_3^1, v_1u_3^2,$  $w_1\lnot u_1^1, w_1\lnot u_1^2, w_1\lnot u_2^1, w_1\lnot u_2^2, w_1\lnot u_3^1$ and $w_1\lnot u_3^2$. This completes the description of $G$. The obtained graph has $6n+2m$ vertices and $8n+12m$ edges.

Let $R$ be a maximum $\rho_3(G)$-set. Without loss of generality we may assume that each vertex $a_i^1$ belongs to $R$ for $i=1,2,\dots, n$. Therefore, $\rho_3(G)\geq n$ and by Proposition~\ref{p3}, $\gamma_{\rm pr}(G)\geq 2\rho_3(G)\geq 2n$. On the other hand, by the construction of $G$ it is not possible to obtain a 3-packing set of cardinality greater than $n$. Therefore, $\rho_3(G)=n$.

First, assume that there exists an instance $E$ such that the three values in each clause are not all equal to each other and consider a satisfying truth assignment. By the construction of $G(u_i)$, every minimum paired dominating set contains at least two vertices of this subgraph. Let $D$ be a paired-dominating set of $G$ such that it contains vertices of $G(u_i)$ that correspond to true literals. For example, if $u_1$ is true in $E$, then $u_1^1, u_1^2\in D$ and if $\lnot u_1$ is true, then $\lnot u_1^1, \lnot u_1^2\in D$. Since the three values in each clause are not all equal to each other, both $v_j$ and $w_j$ are dominated. Therefore, if there exists an instance such that the three values in each clause are not all equal to each other, then $\gamma_{\rm pr}(G)\leq 2n$ implying that $\gp(G)=\rho_3(G)$. 

Assume now that in each instance $E$ there exists a clause in which all three values are equal to each other. Let $C_k$ be such a clause and let $D$ be a minimum paired-dominating set of $G$. Then at least two vertices of each $G(u_i)$ belong to $D$. Let $D$ contain vertices of $G(u_i)$ corresponding to true literals. Then either $v_k$ or $w_k$ is not dominated by true literal vertices. In this situation we need at least two more vertices to dominate $\{v_k\}\cup\{w_k\}$. This implies that $|D|> 2n$. Hence, if in each instance $E$ there exists a clause in which all three values are equal to each other, then $\gamma_{\rm pr}(G)>2n$ and therefore, $\gp(G)>\rho_3(G)$.

\end{proof}

\section{Concluding remarks and open problems}
In this paper, we provide a simple characterization of trees with $\gamma_{pr}(G)=2\gamma(G)$. In addition, we derive NP-hardness results for $\gamma_{pr}(G)=2\gamma(G)$ bipartite graphs, $\gamma_{pr}(G)=\gamma_{t}(G)$  bipartite graphs, $\gamma_{pr}(G)=2\rho_{4}(G)$ graphs, and $\gamma_{pr}(G)=2\rho_{3}(G)$ graphs.

We  remark that $\gamma_{\rm pr}(G)= 2\rho_3(G)$ is not true for block graphs and unicyclic graphs. In fact, the ratio is unbounded in block graphs. Consider a complete graph $K_p$ and add one pendant edge to every vertex of $K_p$. For unicyclic graphs, take a $C_3$ and add pendant edges to every vertex. This renders the extension of our characterization for $\gamma_{pr}(G)=2\gamma(G)$ trees to block graphs and unicyclic graphs difficult. 

Studying inequality chains involving the paired domination number arise many interesting questions which are worth to study. We finish the paper with some open problems.
\begin{enumerate}
\item Characterize $(\gp,2\gamma)$-graphs for graph classes other than trees.
\item Characterize $(\gp,2\rho_3)$-graphs for graph classes other than trees.
\item Study the  $(\gp,2\rho_4)$-graphs. To the best of our knowledge, this question is completely unexplored in the literature.
\end{enumerate}


\begin{thebibliography}{99}

\bibitem{AlvaroRautenbach} J.D. Alvaro, S. Dantas, D. Rautenbach, Perfectly relating the domination, total domination, and paired domination numbers of a graph, {\em Discr. Math} 338 (2015) 1424--1431.


\bibitem{Bresar} B. Bre\v sar, M. A. Henning, and D. F. Rall, Paired-domination of Cartesian product of graphs, {\em Util. Math.} 73 (2007) 255--265.

\bibitem{HHS98} T.W.~Haynes, S.T.~Hedetniemi, and P.J.~Slater,
Fundamentals of Domination in Graphs, Marcel Dekker, New York, 1998.

\bibitem{HHS98b} T.W.~Haynes, S.T.~Hedetniemi, and P.J.~Slater,
Domination in Graphs: Advanced Topics, Marcel Dekker, New York, 1998.

\bibitem{HSlater}T. W. Haynes, P. Slater, Paired--domination in graphs, {\em Networks} 32 (1998) 199--206.

\bibitem{HenningV}M. A. Henning, P. D. Vestergaard, Trees with paired--domination number twice their domination number, Aalborg University, 2006

\bibitem{Hou} Xinmin Hou, A characterization of $(2\gamma, \gamma_{\rm p})$-trees, {\em Discr. Math} 308 (2008) 3420--3426.

\bibitem{GuntherHartnellMarkusRall} G.~Gunther, B.~Hartnell, L.R.~Markus, and D.~Rall,
Graphs with unique minimum dominating  sets, {\em Congr. Numer.} 101 (1994), 55--63.

\bibitem{GJ79} M.R. Garey, D.S. Johnson,  Computers and Intractability: A Guide to the Theory of NP-Completeness (Freeman, San Francisco, 1979).

\bibitem{Moon}A. Meir. J. W. Moon, Relations between packing and covering numbers of a tree, {\em Pacific J. Math.} 61 (1975), 225--233.

\bibitem{Rall}D. F. Rall, Packing and Domination Invariants on Cartesian Products and Direct Products, 2008.

\bibitem{Schaefer} T.J. Schaefer, The complexity of satisfiability problems, Proc. 10th Ann. ACM Symp. on Teory of Computing, Association of Computing Machinery, New York (1978) 216--226.

\bibitem{ShangKangHenning} E. Shang, L. Kang, M. A. Henning, A characterization of trees with equal total domination and paired-domination numbers, {\em Australasian J. of Comb.}  30 (2004) 31--39.
\end{thebibliography}
\end{document}